\documentclass[letterpaper]{article} 
\usepackage{aaai2026}  
\usepackage{times}  
\usepackage{helvet}  
\usepackage{courier}  
\usepackage[hyphens]{url}  
\usepackage{graphicx} 
\usepackage{natbib}  
\usepackage{caption} 
\usepackage{algorithm}

\usepackage{newfloat}
\usepackage{listings}
\DeclareCaptionStyle{ruled}{labelfont=normalfont,labelsep=colon,strut=off} 
\floatstyle{ruled}
\newfloat{listing}{tb}{lst}{}
\floatname{listing}{Listing}
\usepackage{cite}
\usepackage{lscape} 
\usepackage{mathtools}
\usepackage{enumitem}
\usepackage{multirow}
\usepackage{amsthm}
\usepackage{amssymb}
\usepackage{amsmath}
\usepackage{blindtext}
\usepackage{algpseudocode}

\theoremstyle{plain}
\newtheorem{theorem}{Theorem}[section]
\newtheorem{lemma}[theorem]{Lemma}

\newtheorem{proposition}[theorem]{Proposition}

\theoremstyle{definition}

\theoremstyle{remark}

\usepackage{tcolorbox}
\newcommand{\rqq}[1]{
\begin{center}
\begin{tcolorbox}[width=\columnwidth, 
 colback=gray!5!white, 
 colframe=cyan!60!black, 
boxrule=0.5px,
left=2pt,
right=2pt,
top=2pt,
bottom=2pt,
arc=5pt,
auto outer arc]
{#1}
\end{tcolorbox}
\end{center}
}

\title{Area-Optimal Control Strategies for Heterogeneous Multi-Agent Pursuit}
\author{
    Kamal Mammadov and Damith C. Ranasinghe
}
\affiliations {
    School of Computer and Mathematical Sciences, the University of Adelaide \\
    kamal.mammadov@adelaide.edu.au, damith.ranasinghe@adelaide.edu.au.
}

\begin{document}

\maketitle

\begin{abstract}
This paper presents a novel strategy for a multi-agent pursuit-evasion game involving multiple faster pursuers with heterogenous speeds and a single slower evader. We define a geometric region, the evader's safe-reachable set, as the intersection of Apollonius circles derived from each pursuer-evader pair. The capture strategy is formulated as a zero-sum game where the pursuers cooperatively minimize the area of this set, while the evader seeks to maximize it, effectively playing a game of spatial containment. By deriving the analytical gradients of the safe-reachable set's area with respect to agent positions, we obtain closed-form, instantaneous optimal control laws for the heading of each agent. These strategies are computationally efficient, allowing for real-time implementation. Simulations demonstrate that the gradient-based controls effectively steer the pursuers to systematically shrink the evader's safe region, leading to guaranteed capture. This area-minimization approach provides a clear geometric objective for cooperative capture.
\end{abstract}

\section{Introduction}
Multi-agent systems are increasingly central to modern robotics and autonomous operations, with critical applications in contested environments where interactions are inherently adversarial \cite{Wang2025, Manoharan2023}. Scenarios ranging from national security and asset protection \cite{Li20111026, Niknami2022} to logistics and surveillance rely on the coordinated actions of intelligent agents. The formal study of these dynamic conflicts falls under the purview of pursuit-evasion games, a foundational topic in control theory and mathematics \cite{Rusnak2005441}. These games model the strategic interactions between a team of pursuers aiming to capture or contain one or more evaders, who in turn seek to evade capture or achieve a separate objective \cite{Prokopov2013753, Shima2011414}. The vast literature explores a wide variety of formulations, considering factors such as limited observations \cite{Lin20151347} and cooperative team dynamics \cite{Liu20135368}.

Differential game theory provides a powerful mathematical lens for analyzing such conflicts and deriving optimal strategies for all players \cite{Garcia2021675}. This framework has been successfully applied to various multi-agent problems, including cooperative target defence \cite{Liang201958} and scenarios involving a defender aiding a pursuer \cite{Garcia2019652}. Our work leverages this paradigm by formulating the multi-pursuer game as a zero-sum pursuit-evasion game over the area of the safe-reachable set, $A_e(t)$. The agents seek to find the Nash equilibrium of the instantaneous min-max problem on the area's rate of change, $\dot{A}_e(t)$.

\rqq{The primary contribution of this study is the generalization of the area-based control framework to the challenging case of pursuers with heterogeneous speeds. We provide a complete problem formulation and a proof for the area-optimal control strategy for the pursuit-evasion game.}

This is achieved through the novel analytical derivation of the gradients of the safe set's area with respect to each agent's position ($\nabla_{\mathbf{p}_i} A_e$ and $\nabla_{\mathbf{e}} A_e$). Critically, this derivation is accomplished without a closed-form expression for the area itself by applying the Leibniz integral rule to the complex geometry formed by the intersection of Apollonius discs. This result provides a clear and geometrically intuitive interpretation of the optimal Nash equilibrium strategies: each pursuer should move directly toward the centroid of the boundary arc it contributes to the safe set, while the evader's optimal heading is a weighted sum of vectors pointing from its position to the centroids of all active boundary arcs. This provides a computationally tractable, optimal, and decentralized control law for a broad and practical class of pursuit-evasion scenarios.

\section{Related Work}\label{sec:related_work}
Extensive research works on pursuit-evasion and target defence can be broadly categorized into three main areas: i)~game-theoretic analyses that seek optimal strategies, ii)~geometric and control-theoretic methods that focus on robust containment, and iii)~approaches that incorporate real-world complexities such as partial information and environmental obstacles.

Differential game theory is the classical tool for analyzing multi-agent conflicts. A significant portion of this literature focuses on the Target-Attacker-Defender (TAD) game, where optimal state-feedback strategies are derived to determine the outcome \cite{Garcia2019553}. These solutions often involve constructing a ``Barrier" surface that partitions the state space into winning and losing regions for each team \cite{Garcia201868, Mammadov2022120}. Recent work has sought to unify these solutions under more general frameworks, such as the ``Holographic Principle" \cite{Mammadov20232725}, and to derive state-feedback Nash equilibria for games in arbitrary dimensions \cite{Mammadov20201}. Other geometric approaches have been used to find optimal strategies for cooperative defence scenarios \cite{Mammadov20212615}.

An alternative line of work focuses on controlling geometric properties of the system to guarantee capture. This includes methods based on encirclement using robust model predictive control \cite{Wang20211473} and cooperative blocking strategies \cite{Garcia2021360}. Barrier construction has also been explored in related problems like the ``fishing game" \cite{Zha20171409}. Of particular relevance to our work is the area-based control framework proposed by \cite{Shah2019451}. They developed a cooperative policy for pursuers to contain and capture an evader by minimizing the area of the evader's reachable set. Their approach, however, was limited to the case where all pursuers have the same speed as the evader, resulting in a safe set defined by the intersection of half-planes. Our work directly generalizes this seminal contribution by tackling the more complex geometry of Apollonius discs that arises from heterogeneous pursuer speeds. By analytically deriving the area gradients for this more general case, we extend the applicability of area-based control while retaining its elegance and providing formal optimality guarantees.

A third stream of research incorporates more realistic operational constraints. This includes the presence of obstacles, which can fundamentally alter agent strategies and require visibility-based objectives \cite{Bhattacharya20161071}. Another major focus is on games with partial or asymmetric information, where agents operate with uncertainty about their opponents' states \cite{Shishika20218111, Xu20241289, Zhang2021}. These problems often involve sequential decision-making, such as defending against intruders that arrive over time \cite{Pourghorban20231673, Pourghorban2023}, or multi-stage objectives as seen in games like Capture-the-Flag \cite{Garcia20184167}. Our work focuses on the fundamental deterministic, full-information game, providing a foundational optimal solution upon which such complexities can be layered in future research. By providing an analytical and decentralized control law, our framework serves as a critical building block for these more complex scenarios.

\section{Problem Formulation}

We consider a pursuit-evasion game in a two-dimensional planar environment, $\mathbb{R}^2$, involving $N$ pursuers and a single evader. Let $\mathbf{p}_i(t) \in \mathbb{R}^2$ denote the position of pursuer $i$ at time $t$, for $i=1, \dots, N$, and let $\mathbf{e}(t) \in \mathbb{R}^2$ denote the position of the evader. All agents are modeled as simple kinematic point-masses with fixed maximum speeds and instantaneous control over their heading. The state equations governing the agents' motion are given by the following set of first-order ordinary differential equations:
\begin{subequations}
\label{eq:agent_kinematics}
\begin{align}
    \dot{\mathbf{p}}_i(t) &= V_i \mathbf{u}_i(t), \qquad \|\mathbf{u}_i(t)\| \le 1, \label{eq:pursuer_kinematics} \\
    \dot{\mathbf{e}}(t) &= V_e \mathbf{v}(t), \qquad \|\mathbf{v}(t)\| \le 1. \label{eq:evader_kinematics}
\end{align}
\end{subequations}
Here, $V_i$ is the constant maximum speed of pursuer $i$, and $V_e$ is the constant maximum speed of the evader. The vectors $\mathbf{u}_i(t) \in \mathbb{R}^2$ and $\mathbf{v}(t) \in \mathbb{R}^2$ are the unit-norm control inputs representing the instantaneous heading directions for pursuer $i$ and the evader, respectively. A critical assumption for capture to be feasible is that every pursuer is faster than the evader. That is, we assume $V_i > V_e$ for all $i \in \{1, \dots, N\}$. We introduce the speed ratio $\alpha_{i} := V_{e}/V_{i}\in(0,1)$ to quantize the speed advantage of pursuer $i$ relative to the evader. The pursuers' objective is to cooperate to capture the evader, while the evader's objective is to take evasive action to avoid capture.\footnote{The safe-reachable set defines the optimal pursuit strategy, as opposed to minimum capture time \cite{von2019multi}.}

To quantify the state of the game from a geometric perspective, we first analyze the one-on-one engagement between a single pursuer $i$ and the evader. The key question is: which points in the plane can the evader reach before, or at the same time as, pursuer $i$? The locus of such points defines a region of temporary safety for the evader with respect to that specific pursuer. Due to the constant speed assumption ($V_i > V_e$), this region is a circle, known as the Apollonius disc. Any point inside this disc is reachable by the evader strictly before the pursuer.

The center $\mathbf{c}_i(t)$ and radius $r_i(t)$ of the Apollonius disc for the pursuer-evader pair $(i,e)$ are functions of their current positions and speeds. They are given by:
\begin{subequations}
\label{eq:apollonius_i}
\begin{align}
    \mathbf{c}_i(t) &= \frac{V_i^2}{V_i^2 - V_e^2}\mathbf{e}(t) + \frac{V_e^2}{V_e^2 - V_i^2}\mathbf{p}_i(t), \label{eq:ap_circle_center} \\
    r_i(t) &= \frac{V_i V_e}{|V_i^2 - V_e^2|} \|\mathbf{e}(t) - \mathbf{p}_i(t)\|. \label{eq:ap_circle_radius}
\end{align}
\end{subequations}

While the Apollonius disc defines safety against a single pursuer, the evader must be safe from all $N$ pursuers simultaneously. Therefore, the true region of safety for the evader is the set of all points that are inside every individual Apollonius disc. We define this region as the evader's instantaneous \textit{safe-reachable set}, denoted $\mathcal{S}_e(t)$. It is formally defined as the intersection of the $N$ Apollonius discs:
\begin{equation}
\label{eq:srs_def}
\mathcal{S}_e(t) = \bigcap_{i=1}^{N} \Bigl\{\mathbf{q} \in \mathbb{R}^2 \;\bigl|\; \|\mathbf{q} - \mathbf{c}_i(t)\| \le r_i(t)\Bigr\}.
\end{equation}
This convex set represents the evader's instantaneous `safe zone', containing all locations the evader can reach before any pursuer can intercept it. If this set shrinks to a single point or becomes empty, capture is imminent. Thus the collective goal for the pursuer team is to cooperatively manoeuvrer to shrink this set. We quantify this objective using the area of the set, denoted by $A_e(t)$:
\begin{equation}
\label{eq:area_def}
A_e(t) = \operatorname{Area}\bigl(\mathcal{S}_e(t)\bigr).
\end{equation}
As $A_e(t) \to 0$, the evader is left with no safe space, and capture becomes inevitable. The game is framed around the reduction of this area.

\section{Area-Based Optimal Control Strategy}
The core of our proposed guidance law is to actively control the area of the evader's safe-reachable set, $A_e(t)$. To develop a strategy that systematically reduces this area, we must first analyze its rate of change. The time derivative of the area, $\dot{A}_e(t)$, can be found by applying the chain rule, as the area is an implicit function of the positions of all agents, which are themselves functions of time.

First, we define the gradients of the area with respect to the agent positions. Let $\mathbf{F}_{\mathbf{p}_i}(t)$ and $\mathbf{F}_{\mathbf{e}}(t)$ be the gradients with respect to the position of pursuer $i$ and the evader, respectively:
\begin{subequations}
\label{eq:flux_terms}
\begin{align}
    \mathbf{F}_{\mathbf{p}_i}(t) &:= \nabla_{\mathbf{p}_i} A_e(t), \\
    \mathbf{F}_{\mathbf{e}}(t) &:= \nabla_{\mathbf{e}} A_e(t).
\end{align}
\end{subequations}
Using these gradient vectors, the total time derivative of the area is given by the sum of the inner products of each gradient with the corresponding agent's velocity:
\begin{equation}
    \dot{A}_e(t) = \sum_{i=1}^{N} \mathbf{F}_{\mathbf{p}_i}(t)^{\top} \dot{\mathbf{p}}_i(t) + \mathbf{F}_{\mathbf{e}}(t)^{\top} \dot{\mathbf{e}}(t).
\end{equation}
Substituting the agent kinematics from Eq.~\eqref{eq:agent_kinematics}, we arrive at an expression for the rate of change of the area in terms of the agents' control inputs:
\begin{equation}
\label{eq:area_derivative}
    \dot{A}_e(t) = \sum_{i=1}^{N} V_i \mathbf{F}_{\mathbf{p}_i}(t)^{\top} \mathbf{u}_i(t) + V_e \mathbf{F}_{\mathbf{e}}(t)^{\top} \mathbf{v}(t).
\end{equation}
This equation is fundamental to our control strategy, as it directly links the actions of the agents ($\mathbf{u}_i, \mathbf{v}$) to the change in the area of the safe-reachable set.

The control problem can now be framed as a zero-sum game played over the rate of change of the area, $\dot{A}_e(t)$. The two opposing sides have conflicting objectives. The team of pursuers collectively chooses their heading controls, $\{\mathbf{u}_i\}_{i=1}^N$, to make $\dot{A}_e(t)$ as negative as possible, thereby shrinking the safe-reachable set at the maximum possible rate. Simultaneously, the evader chooses its heading control, $\mathbf{v}$, to make $\dot{A}_e(t)$ as large as possible, aiming to expand or maintain its safe zone. This instantaneous, greedy approach leads to the following min-max optimization problem:
\begin{equation}
\label{eq:minmax_game}
\min_{\{\mathbf{u}_i\}} \max_{\mathbf{v}} \dot{A}_e(t).
\end{equation}
The solution to this game provides the instantaneous optimal heading controls for all agents.

The optimal solution to the min-max game in Eq.~\eqref{eq:minmax_game} can be found by recognizing that the pursuers' and evader's controls are decoupled. Each agent seeks to optimize its own term in the sum for $\dot{A}_e(t)$. The optimal strategy to minimize (or maximize) the inner product of a fixed vector (the gradient), and a unit control vector, is to align the control vector with the negative (or positive) direction of the fixed vector. This principle gives rise to the closed-form optimal control laws.

\setcounter{section}{1}
\begin{lemma}\label{Nash_equilibrium}
Given the area gradients $\mathbf{F}_{\mathbf{p}_i}(t)$ and $\mathbf{F}_{\mathbf{e}}(t)$, the optimal headings $\mathbf{u}_i^\star(t)$ and $\mathbf{v}^\star(t)$ that solve the min-max game \eqref{eq:minmax_game} are given by:
\begin{subequations}
\label{eq:opt_controls}
\begin{align}
    \mathbf{u}_i^\star(t) &= -\frac{\mathbf{F}_{\mathbf{p}_i}(t)}{\|\mathbf{F}_{\mathbf{p}_i}(t)\|}, \quad i=1,\dots,N, \label{eq:opt_ui} \\[2mm]
    \mathbf{v}^\star(t) &= \frac{\mathbf{F}_{\mathbf{e}}(t)}{\|\mathbf{F}_{\mathbf{e}}(t)\|}. \label{eq:opt_ve}
\end{align}
\end{subequations}
\end{lemma}

\begin{proof}
The expression for $\dot{A}_e(t)$ in Eq.~\eqref{eq:area_derivative} is a sum of terms, where each pursuer's control $\mathbf{u}_i$ and the evader's control $\mathbf{v}$ appear independently. To solve the min-max problem, we can optimize each term separately.

For each pursuer $i$, the objective is to minimize its contribution to the sum, which is the term $V_i \mathbf{F}_{\mathbf{p}_i}(t)^{\top} \mathbf{u}_i(t)$. Since $V_i > 0$, this is equivalent to minimizing the inner product $\mathbf{F}_{\mathbf{p}_i}(t)^{\top} \mathbf{u}_i(t)$. The inner product between two vectors is minimized when they point in opposite directions. Therefore, the optimal unit-norm control for pursuer $i$ is to align its heading with the negative gradient direction, yielding $\mathbf{u}_i^\star(t) = -\mathbf{F}_{\mathbf{p}_i}(t) / \|\mathbf{F}_{\mathbf{p}_i}(t)\|$.

Conversely, the evader's objective is to maximize its contribution, $V_e \mathbf{F}_{\mathbf{e}}(t)^{\top} \mathbf{v}(t)$. This is equivalent to maximizing the inner product $\mathbf{F}_{\mathbf{e}}(t)^{\top} \mathbf{v}(t)$. The inner product is maximized when the two vectors are aligned. Thus, the optimal unit-norm control for the evader is to align its heading with the gradient direction, $\mathbf{v}^\star(t) = \mathbf{F}_{\mathbf{e}}(t) / \|\mathbf{F}_{\mathbf{e}}(t)\|$.
\end{proof}

When all agents employ their optimal strategies as defined in Lemma \ref{Nash_equilibrium}, we can determine the resulting dynamics of the area $A_e(t)$. Substituting the optimal controls from Eq.~\eqref{eq:opt_controls} back into the expression for $\dot{A}_e(t)$ in Eq.~\eqref{eq:area_derivative} yields the time evolution of the area under the optimal policy:
{\allowdisplaybreaks
\begin{align*}
\dot{A}_e(t) &= \sum_{i=1}^{N} V_i \mathbf{F}_{\mathbf{p}_i}(t)^{\top} \mathbf{u}_i^\star(t) + V_e \mathbf{F}_{\mathbf{e}}(t)^{\top} \mathbf{v}^\star(t) \nonumber \\
&= \sum_{i=1}^{N} V_i \mathbf{F}_{\mathbf{p}_i}(t)^{\top} \left(-\frac{\mathbf{F}_{\mathbf{p}_i}(t)}{\|\mathbf{F}_{\mathbf{p}_i}(t)\|}\right)\\ &+ V_e \mathbf{F}_{\mathbf{e}}(t)^{\top} \left(\frac{\mathbf{F}_{\mathbf{e}}(t)}{\|\mathbf{F}_{\mathbf{e}}(t)\|}\right) \nonumber \\
&= V_e \|\mathbf{F}_{\mathbf{e}}(t)\| - \sum_{i=1}^{N} V_i \|\mathbf{F}_{\mathbf{p}_i}(t)\|. 
\tag{\stepcounter{equation}\theequation}\\\label{eq:optimal_area_rate}
\end{align*}
}
This elegant result reveals that the game becomes a direct competition between the weighted magnitudes of the area gradients. The area increases based on the evader's ability to move in its most effective direction (scaled by $\|\mathbf{F}_{\mathbf{e}}\|$), and decreases based on the pursuers' collective ability to move in their most effective directions (scaled by the sum of $V_i\|\mathbf{F}_{\mathbf{p}_i}\|$). The sign of $\dot{A}_e(t)$ is determined by which side can generate a larger aggregate ``force" on the area's boundary.

The framework presented thus far provides a clear Nash equilibrium strategy, contingent on one crucial component: the area gradients themselves. The main remaining challenge is to compute the vectors $\nabla_{\mathbf{p}_i} A_e(t)$ and $\nabla_{\mathbf{e}} A_e(t)$. Because the area $A_e(t)$ is a complex geometric quantity, the intersection of multiple discs, its gradients are not trivial to compute. The following section is dedicated to the analytical derivation of these gradients, which is the final piece required to make the control strategy fully implementable.

\section{Analytical Gradient Derivation}
The optimal control strategy derived in the previous section is entirely dependent on the area gradients, $\nabla_{\mathbf{p}_i} A_e(t)$ and $\nabla_{\mathbf{e}} A_e(t)$. This section presents the main technical contribution of this paper: the analytical derivation of these gradients. A standard calculus-based approach would first require deriving a closed-form analytical expression for the area $A_e(t)$ as a function of the centers and radii of the constituent Apollonius discs. Subsequently, one would compute the partial derivatives of this area function with respect to each center and radius, and finally apply the multivariable chain rule to find the gradients with respect to the agent positions $\mathbf{p}_i$ and $\mathbf{e}$. However, deriving an analytical formula for the intersection area of an arbitrary number of discs is an extremely sophisticated geometric problem, yielding highly complex expressions, if they are achievable at all. Consequently, differentiating such functions to find the gradients becomes essentially intractable for $N>2$.

Fortunately, an alternative and more elegant path exists. The related prior publication \cite{Shah2019451} demonstrated the application of the Leibniz integral rule (also known as Reynolds transport theorem) to derive these gradients without requiring a closed-form expression for the area itself. To understand this intuitively, we note that the optimal policy only depends on the instantaneous rate of change of the area, $\dot{A}_e(t)$. Therefore, it is not necessary to know the value of $A_e(t)$; it suffices to know how the boundary of the safe-reachable set $\mathcal{S}_e(t)$ moves in response to changes in the agents' positions. The Leibniz integral rule provides precisely this: an analytical method for determining the rate of change of an area by integrating the normal velocity of its boundary, without needing an upfront formula for the area.

In this section, we apply a coordinate-independent formulation of the Leibniz integral rule. The key insight is that the change in area can be computed by analyzing the derivatives of the level set functions that implicitly define the boundaries of the Apollonius discs. This powerful technique allows us to bypass not only the need for an area formula but also the need to explicitly parameterize the boundary of the intersection region $\mathcal{S}_e(t)$. Calculating the derivatives of the level set functions with respect to the agent positions is sufficient to determine the area gradients.

\subsection{Boundary of the Safe-Reachable Set}
Since the change in the area $A_e(t)$ is related to the movement of its boundary, the first step in our derivation is to obtain an explicit parametrization of this boundary, denoted $\partial\mathcal{S}_e(t)$. The safe-reachable set $\mathcal{S}_e(t)$ is a convex region formed by the intersection of $N$ Apollonius discs. Its boundary, therefore, is composed of a sequence of circular arcs, where each arc belongs to one of the Apollonius circles. An arc from a given circle contributes to the boundary only if it lies inside all other discs. The following analysis will formalize this geometric construction.

Let $\mathcal{B}_i(t)$ denote the closed Apollonius disc for pursuer $i$ at time $t$, defined as $\mathcal{B}_i(t) = \{\mathbf{q} \in \mathbb{R}^2 \mid \|\mathbf{q} - \mathbf{c}_i(t)\| \le r_i(t)\}$. The safe-reachable set from Eq.~\eqref{eq:srs_def} is then the intersection of these $N$ discs:
\begin{equation}
    \mathcal{S}_e(t) = \bigcap_{i=1}^{N} \mathcal{B}_i(t).
\end{equation}
By construction, each disc $\mathcal{B}_i(t)$ is a compact and convex set. Since the intersection of any number of compact and convex sets is also compact and convex, $\mathcal{S}_e(t)$ is a non-empty (by assumption of a non-trivial game state) compact convex subset of $\mathbb{R}^2$.

A preliminary reduction simplifies the problem by eliminating discs whose boundaries cannot contribute to $\partial\mathcal{S}_e(t)$. If a disc $\mathcal{B}_i(t)$ wholly contains another disc $\mathcal{B}_j(t)$, then $\mathcal{B}_i(t)$'s boundary lies entirely outside $\mathcal{B}_j(t)$ and thus cannot form part of $\partial\mathcal{S}_e(t)$. We therefore restrict attention to discs that do not fully contain any other single disc. Define the active index set $I(t)$ as:
\begin{align*}
I(t) =\\ &\Bigl\{ i \in \{1, \dots, N\} \mid \|\mathbf{c}_i(t) - \mathbf{c}_j(t)\| \\&\geq r_i(t) - r_j(t) \quad \forall j \neq i \Bigr\}.
\end{align*}
The condition $|\mathbf{c}_i(t) - \mathbf{c}_j(t)| + r_j(t) \geq r_i(t)$ ensures $\mathcal{B}_i(t)$ does not fully contain any other disc $\mathcal{B}_j(t)$. Consequently, the boundary of the safe-reachable set is composed entirely of arcs from the circles $\{\partial\mathcal{B}_i(t) \mid i \in I(t)\}$. All subsequent analysis considers only discs in this active set.

To find the specific arcs forming the boundary, we must determine, for each active circle $\partial\mathcal{B}_i(t)$, which parts are intersected by all other active discs $\mathcal{B}_j(t)$. This leads to the following lemma, which provides an explicit algorithm for computing the angular range of each contributing arc.

\begin{lemma}[Boundary Arc Computation]\label{lemma_Boundary_Arc_Computation}
For each active circle $i \in I(t)$, the set of angles $\Theta_i(t)$ corresponding to the arc on $\partial\mathcal{B}_i(t)$ that forms part of $\partial\mathcal{S}_e(t)$ is a single compact interval $[\theta_i^{\min}, \theta_i^{\max}]$ or the empty set. This interval is computed via the following steps:
\begin{enumerate}
    \item For each other active circle $j \in I(t)$ where $j \neq i$, calculate the pairwise constraint interval $\Gamma_{ij}(t) = [\phi_{ij} - \alpha_{ij}, \phi_{ij} + \alpha_{ij}]$ (mod $2\pi$), where the angles are given by:
    \begin{align}
        \phi_{ij} &\text{ is the angle of the vector from $\mathbf{c}_i(t)$ to $\mathbf{c}_j(t)$}, \\
        \alpha_{ij} &= \arccos\left(\frac{\|\mathbf{c}_j(t) - \mathbf{c}_i(t)\|^2 + r_i(t)^2 - r_j(t)^2}{2 r_i(t) \|\mathbf{c}_j(t) - \mathbf{c}_i(t)\|}\right).
    \end{align}
    \item The final angular range for circle $i$ is the intersection of all such constraint intervals:
    \begin{equation}
        \Theta_i(t) = \bigcap_{j \in I(t), j \neq i} \Gamma_{ij}(t).
    \end{equation}
\end{enumerate}
If $\Theta_i(t)$ is non-empty, its endpoints define $\theta_i^{\min}$ and $\theta_i^{\max}$.
\end{lemma}
Proof contained in the Technical Appendix. This lemma provides a complete algorithm to compute the boundary arcs of the safe-reachable set at any time $t$. For notational simplicity, the index set $I(t)$ will henceforth refer exclusively to circles that contribute a non-empty arc to the boundary, meaning those indices $i$ for which $\Theta_i(t) \neq \varnothing$.

\begin{figure}[t]
\centering
\includegraphics[width=1\columnwidth]{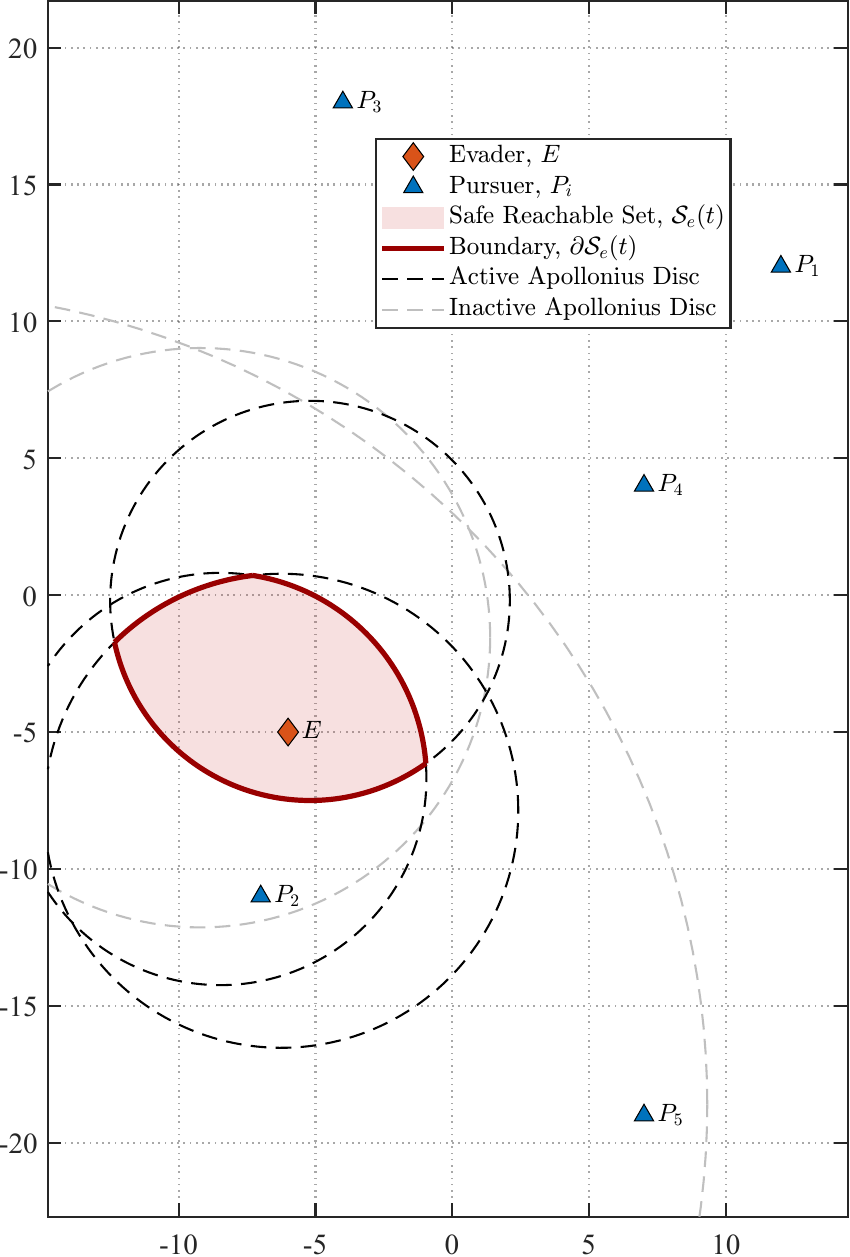}
\caption{Visualizing the Safe-Reachable Set $\mathcal{S}_e(t)$ for the example scenario with five pursuers and a single evader. The evader's speed is $V_e = 4$, while the pursuer speeds are given by the vector $V = [6, 6, 12, 10, 9]$. Only pursuers 2, 3 and 4 are active and contribute to the boundary.}
\label{fig:safe_set}
\end{figure}

\subsection{Leibniz Integral Rule}

With the boundary of the safe-reachable set $\mathcal{S}_e(t)$ explicitly parameterized, we now turn to the core task of computing its area gradients. As motivated previously, we use the Leibniz integral rule to find the derivative of the area with respect to the agent positions without requiring a closed-form expression for the area itself. This is possible because the change in area of a moving region is determined by the normal velocity of its boundary.

First, we define each Apollonius disc $\mathcal{B}_i(t)$ implicitly using a level-set function, $F_i: \mathbb{R}^2 \to \mathbb{R}$, such that:
\begin{equation}
    \mathcal{B}_i(t) = \bigl\{\mathbf{q} \in \mathbb{R}^2 \mid F_i(\mathbf{q}, t) \le 0 \bigr\},
\end{equation}
where the level-set function is given by:
\begin{equation}
    F_i(\mathbf{q}, t) = \|\mathbf{q} - \mathbf{c}_i(t)\|^2 - r_i(t)^2.
\end{equation}
The boundary of the disc, $\partial\mathcal{B}_i(t)$, is the zero level-set, $F_i(\mathbf{q}, t) = 0$.

The Leibniz integral rule states that the derivative of the area of the region $\mathcal{S}_e(t)$ with respect to any parameter $\zeta$ (such as a component of an agent's position vector) is given by an integral over its boundary $\partial\mathcal{S}_e(t)$:
\begin{equation}
    \frac{\partial A_e(t)}{\partial \zeta} = \int_{\partial\mathcal{S}_e(t)} V_n(\mathbf{q}, \zeta) \,ds,
\end{equation}
where $V_n$ is the outward normal velocity of the boundary at point $\mathbf{q}$ induced by a change in the parameter $\zeta$, and $ds$ is the element of arc length. The key insight from \cite{Shah2019451} is that this normal velocity can be expressed in terms of the derivatives of the level-set function that defines the boundary at that point. If a point $\mathbf{q}$ on the boundary belongs to the arc from circle $i$, its normal velocity is given by:
\begin{equation}
    V_n(\mathbf{q}, \zeta) = -\frac{\partial F_i / \partial \zeta}{\|\nabla_{\mathbf{q}} F_i\|}.
\end{equation}
This formula allows us to compute the change in area by integrating the effects of parameter changes on the implicit functions that define the boundary, a much more tractable approach than direct differentiation of a complex area formula.

To apply the Leibniz rule, we must compute the gradients of $F_i(\mathbf{q}, t)$ with respect to the spatial variable $\mathbf{q}$ and the agent positions $\{\mathbf{p}_k\}_{k=1}^N$ and $\mathbf{e}$. The function $F_i$ depends on the agent positions implicitly through the Apollonius center $\mathbf{c}_i$ and radius $r_i$. 
\begin{proposition}[Level-Set Gradients]\label{prop:level_set_gradients}
The gradients of the level-set function $F_i(\mathbf{q},t)$ with respect to the spatial variable $\mathbf{q}$ and the agent positions are given by:
\begin{align}
    \nabla_{\mathbf{q}} F_i(\mathbf{q},t) &= 2(\mathbf{q} - \mathbf{c}_i), \\
    \nabla_{\mathbf{p}_i} F_i(\mathbf{q},t) &= \frac{2\alpha_i^2}{1-\alpha_i^2}(\mathbf{q} - \mathbf{p}_i), \\
    \nabla_{\mathbf{e}} F_i(\mathbf{q},t) &= \frac{2}{1-\alpha_i^2}(\mathbf{e} - \mathbf{q}), \\
    \nabla_{\mathbf{p}_k} F_i(\mathbf{q},t) &= \mathbf{0}, \quad \text{for } k \neq i.
\end{align}
\end{proposition}
Proof contained in the Appendix. These expressions provide all the necessary components to evaluate the Leibniz integral.

\subsubsection{Area Gradient Calculation}\label{sectionAreaGradientCalculation}
We can now assemble these components to compute the area gradients. The boundary $\partial\mathcal{S}_e(t)$ is a union of arcs from the active circles $i \in I(t)$, where each arc is parameterized by $\theta \in [\theta_i^{\min}(t), \theta_i^{\max}(t)]$. The Leibniz integral becomes a sum of integrals over these arcs:
\begin{equation}
    \nabla_{\mathbf{z}} A_e(t) = \sum_{i \in I(t)} \int_{\theta_i^{\min}(t)}^{\theta_i^{\max}(t)} -\frac{\nabla_{\mathbf{z}} F_i(\mathbf{q}(\theta),t)}{\|\nabla_{\mathbf{q}} F_i(\mathbf{q}(\theta),t)\|} \, r_i \,d\theta,
\end{equation}
where $\mathbf{z}$ is any agent position vector and $ds = r_i d\theta$. On the boundary arc of circle $i$, we have $\|\mathbf{q}(\theta) - \mathbf{c}_i\| = r_i$, so the denominator becomes $\|\nabla_{\mathbf{q}} F_i\| = \|2(\mathbf{q} - \mathbf{c}_i)\| = 2r_i$.

For the gradient with respect to a pursuer position $\mathbf{p}_k$, we note that $\nabla_{\mathbf{p}_k} F_i = \mathbf{0}$ if $k \neq i$. Thus, only the term for $i=k$ contributes:
\begin{align}
    \nabla_{\mathbf{p}_k} A_e(t) &= \int_{\theta_k^{\min}(t)}^{\theta_k^{\max}(t)} -\frac{\nabla_{\mathbf{p}_k} F_k(\mathbf{q}(\theta),t)}{2r_k} \, r_k \,d\theta \nonumber \\
    &= -\frac{1}{2} \int_{\theta_k^{\min}(t)}^{\theta_k^{\max}(t)} \frac{2\alpha_k^2}{1-\alpha_k^2}(\mathbf{q}(\theta) - \mathbf{p}_k) \,d\theta \nonumber \\
    &= -\frac{\alpha_k^2}{1-\alpha_k^2} \int_{\theta_k^{\min}(t)}^{\theta_k^{\max}(t)} (\mathbf{q}(\theta) - \mathbf{p}_k) \,d\theta. \label{pursuer_integral_rule}
\end{align}
For the gradient with respect to the evader position $\mathbf{e}$, all active arcs contribute:
\begin{align}
    \nabla_{\mathbf{e}} A_e(t) &= \sum_{i \in I(t)} \int_{\theta_i^{\min}(t)}^{\theta_i^{\max}(t)} -\frac{\nabla_{\mathbf{e}} F_i(\mathbf{q}(\theta),t)}{2r_i} \, r_i \,d\theta \nonumber \\
    &= -\frac{1}{2} \sum_{i \in I(t)} \int_{\theta_i^{\min}(t)}^{\theta_i^{\max}(t)} \frac{2}{1-\alpha_i^2}(\mathbf{e} - \mathbf{q}(\theta)) \,d\theta \nonumber \\
    &= -\sum_{i \in I(t)} \frac{1}{1-\alpha_i^2} \int_{\theta_i^{\min}(t)}^{\theta_i^{\max}(t)} (\mathbf{e} - \mathbf{q}(\theta)) \,d\theta. \label{evader_integral_rule}
\end{align}
The problem now reduces to evaluating these vector-valued integrals over the angular intervals $[\theta_i^{\min}, \theta_i^{\max}]$. The integrands $(\mathbf{q}(\theta) - \mathbf{p}_k)$ and $(\mathbf{e} - \mathbf{q}(\theta))$ can be expressed in terms of $\theta$ and agent positions. After performing the integration and simplifying using trigonometric identities, we arrive at a coordinate-independent geometric form for the gradients.

\begin{theorem}[Geometric Form of the Area Gradients]\label{thm:area_gradients}
Let $\mathbf{C}_{\text{arc},i}(t)$ denote the geometric centroid of the arc on the boundary $\partial\mathcal{S}_e(t)$ contributed by an active pursuer $i \in I(t)$. The centroid is the average position of all points on the arc, defined by the line integral:
\begin{equation}
    \mathbf{C}_{\text{arc},i}(t) \triangleq \frac{1}{L_i(t)} \int_{\text{arc}_i} \mathbf{q}(s) \, ds,
\end{equation}
where $L_i(t)$ is the length of the boundary arc contributed by circle $i$ and $\mathbf{q}(s)$ is its position vector parametrized by arc length $s$. 

The area gradients with respect to the agent positions are given in terms of these centroids, yielding direct geometric interpretations for the optimal headings:
\begin{enumerate}
    \item \textbf{Pursuers:} The area gradient with respect to the position of any active pursuer $i \in I(t)$ is
    \begin{equation}\label{p_i_geometric}
        \nabla_{\mathbf{p}_i} A_e(t) = -  (1-\alpha_i^2) r_i(t) L_i(t) \left( \mathbf{C}_{\text{arc},i}(t) - \mathbf{p}_i(t) \right).
    \end{equation}
    Since the optimal heading $\mathbf{u}_i^\star(t)$ is anti-parallel to $\nabla_{\mathbf{p}_i} A_e(t)$, the pursuer's optimal strategy is to move directly towards the centroid $\mathbf{C}_{\text{arc},i}(t)$ of its corresponding boundary arc.

    \item \textbf{Evader:} The area gradient with respect to the evader's position is
    \begin{equation}\label{e_geometric}
        \nabla_{\mathbf{e}} A_e(t) = \sum_{i \in I(t)} \frac{1-\alpha_i^2}{\alpha_i^2}  r_i(t)L_i(t) \left( \mathbf{C}_{\text{arc},i}(t) - \mathbf{e}(t) \right).
    \end{equation}
    Since the optimal heading $\mathbf{v}^\star(t)$ is parallel to $\nabla_{\mathbf{e}} A_e(t)$, the evader's optimal strategy is to move in the direction of a weighted sum of vectors. Each vector in the sum, points from the evader's current position to the centroid of an active boundary arc, $\mathbf{C}_{\text{arc},i}(t)$.
\end{enumerate}

The centroid of each active arc $i \in I(t)$ can be computed explicitly by:
\begin{equation}
    \mathbf{C}_{\text{arc},i}(t) = \mathbf{c}_i(t) + r_i(t)\left( \frac{\ell_i(t)}{L_i(t)} \right) \mathbf{u}_i(t),
\end{equation}
where the terms are defined by the geometry of the Apollonius circle and its active arc:
\begin{itemize}
    \item[\textbullet] $\mathbf{c}_i(t) = \frac{\mathbf{e}(t) - \alpha_i^2 \mathbf{p}_i(t)}{1-\alpha_i^2}$ and $r_i(t) = \frac{\alpha_i}{1-\alpha_i^2}\|\mathbf{e}(t) - \mathbf{p}_i(t)\|$ are the center and radius of the Apollonius circle.
    \item[\textbullet] $L_i(t) = r_i(t) \Delta\theta_i$ is the arc length, where $\Delta\theta_i = \theta_i^{\max}(t) - \theta_i^{\min}(t)$.
    \item[\textbullet] $\ell_i(t) = 2r_i(t) \sin(\Delta\theta_i/2)$ is the length of the chord connecting the endpoints of that arc.
    \item[\textbullet] $\mathbf{u}_i(t) = (\cos m_i(t), \sin m_i(t))^\top$ is the unit vector pointing from the center $\mathbf{c}_i(t)$ to the midpoint of the chord, aligned with the angle $m_i(t) = (\theta_i^{\max}(t) + \theta_i^{\min}(t))/2$.
\end{itemize}
\end{theorem}
\begin{proof}
The proof establishes a direct correspondence between the integral definitions of the relevant geometric centroids and the integral forms of the area gradients derived earlier. We prove the results for the pursuers and the evader separately.

\paragraph{Pursuer's Optimal Heading}
Let $\mathbf{C}_{\text{arc},i}(t)$ denote the geometric centroid of the arc contributed by pursuer $i$. By definition, the centroid is the average position of all points on the arc. Using a line integral, this is:
\begin{equation*}
    \mathbf{C}_{\text{arc},i}(t) = \frac{1}{L_i(t)} \int_{\text{arc}_i} \mathbf{q}(s) \, ds,
\end{equation*}
where $L_i(t)$ is the length of the arc. Parametrizing by the angle $\theta$, with $ds = r_i d\theta$ and $L_i(t) = r_i \Delta\theta_i$, we obtain:
\begin{equation*}
    \mathbf{C}_{\text{arc},i}(t) = \frac{1}{r_i \Delta\theta_i} \int_{\theta_i^{\min}}^{\theta_i^{\max}} \mathbf{q}(\theta) \, r_i d\theta = \frac{1}{\Delta\theta_i} \int_{\theta_i^{\min}}^{\theta_i^{\max}} \mathbf{q}(\theta) \, d\theta.
\end{equation*}
The vector pointing from the pursuer's position $\mathbf{p}_i(t)$ to this centroid is $\mathbf{V}_{\text{target},i}(t) = \mathbf{C}_{\text{arc},i}(t) - \mathbf{p}_i(t)$. Since $\mathbf{p}_i(t)$ is constant with respect to the integration variable $\theta$, we can write:
\begin{align*}
    \mathbf{V}_{\text{target},i}(t) &= \frac{1}{\Delta\theta_i} \int_{\theta_i^{\min}}^{\theta_i^{\max}} \mathbf{q}(\theta) \, d\theta - \frac{1}{\Delta\theta_i} \int_{\theta_i^{\min}}^{\theta_i^{\max}} \mathbf{p}_i(t) \, d\theta \\
    &= \frac{1}{\Delta\theta_i} \int_{\theta_i^{\min}}^{\theta_i^{\max}} \bigl(\mathbf{q}(\theta) - \mathbf{p}_i(t)\bigr) \, d\theta. 
\end{align*}
Now, we recall the intermediate result \eqref{pursuer_integral_rule} from the area gradient calculation subsection:
\begin{equation*}
    \nabla_{\mathbf{p}_i} A_e(t) = -\frac{\alpha_i^2}{1-\alpha_i^2} \int_{\theta_i^{\min}(t)}^{\theta_i^{\max}(t)} \bigl(\mathbf{q}(\theta) - \mathbf{p}_i(t)\bigr) \,d\theta. \label{eq:gradient_integral_pursuer}
\end{equation*}
By comparing these results, we establish a direct proportionality:
\begin{equation*}
    \nabla_{\mathbf{p}_i} A_e(t) = -\frac{\alpha_i^2 \Delta\theta_i}{1-\alpha_i^2} \mathbf{V}_{\text{target},i}(t).
\end{equation*}
The optimal pursuer heading $\mathbf{u}_i^\star(t)$ is anti-parallel to $\nabla_{\mathbf{p}_i} A_e(t)$ (since pursuers minimize the area). As $\alpha_i \in (0,1)$ and $\Delta\theta_i > 0$ for an active pursuer, the scalar multiple $\frac{\alpha_i^2 \Delta\theta_i}{1-\alpha_i^2}$ is strictly positive. Therefore, the optimal heading $\mathbf{u}_i^\star(t)$ is parallel to and points in the same direction as $\mathbf{V}_{\text{target},i}(t)$, the vector from the pursuer to the arc centroid.

\paragraph{Evader's Optimal Heading}
The geometric interpretation for the evader's heading follows from a similar analysis. We start with the integral form of the evader's area gradient derived from the Leibniz rule in Eq.~\eqref{evader_integral_rule}:
\begin{equation*}
    \nabla_{\mathbf{e}} A_e(t) = -\sum_{i \in I(t)} \frac{1}{1-\alpha_i^2} \int_{\theta_i^{\min}(t)}^{\theta_i^{\max}(t)} (\mathbf{e}(t) - \mathbf{q}(\theta)) \,d\theta 
\end{equation*}
For each active arc $i \in I(t)$, we define a target vector $\mathbf{V}_{\text{target},e,i}(t)$ pointing from the evader's position $\mathbf{e}(t)$ to the centroid of that arc, $\mathbf{C}_{\text{arc},i}(t)$. Using the same integral definition of the centroid as in the pursuer's case, this vector is:
\begin{align*}
    \mathbf{V}_{\text{target},e,i}(t) &= \mathbf{C}_{\text{arc},i}(t) - \mathbf{e}(t) \\
    &= \frac{1}{\Delta\theta_i} \int_{\theta_i^{\min}}^{\theta_i^{\max}} \mathbf{q}(\theta) \, d\theta - \frac{1}{\Delta\theta_i} \int_{\theta_i^{\min}}^{\theta_i^{\max}} \mathbf{e}(t) \, d\theta \\
    &= \frac{1}{\Delta\theta_i} \int_{\theta_i^{\min}}^{\theta_i^{\max}} \bigl(\mathbf{q}(\theta) - \mathbf{e}(t)\bigr) \, d\theta.
\end{align*}
This provides a direct link between the integral in the gradient expression and the vector to the arc centroid. Substituting this relationship into the formula for the area gradient gives:
\begin{equation*}
    \nabla_{\mathbf{e}} A_e(t) = \sum_{i \in I(t)} \left(\frac{\Delta\theta_i}{1-\alpha_i^2}\right) \mathbf{V}_{\text{target},e,i}(t).
\end{equation*}
The optimal evader heading $\mathbf{u}_e^\star(t)$ is parallel to $\nabla_{\mathbf{e}} A_e(t)$ (since the evader maximizes the area). The expression above shows that this gradient is a vector sum, where each term is a vector pointing from the evader to the centroid of an active boundary arc. For any active arc, $\Delta\theta_i > 0$ and $\alpha_i \in (0,1)$, which ensures that all weighting coefficients $\frac{\Delta\theta_i}{1-\alpha_i^2}$ are strictly positive. 

The precise formulation for $\nabla_{\mathbf{p}_i} A_e(t)$ and $\nabla_{\mathbf{e}} A_e(t)$ found in \eqref{p_i_geometric} and \eqref{e_geometric} is obtained with $\Delta\theta_i = L_i/r_i$ and $r_i = \frac{\alpha_i}{1-\alpha_i^2}\|\mathbf{e}-\mathbf{p}_i\|$. The derivation of the explicit formula for the centroid $\mathbf{C}_{\text{arc},i}(t)$ is provided in the Appendix.
\end{proof}

\begin{figure}[ht]
\centering
\includegraphics[width=1\columnwidth]{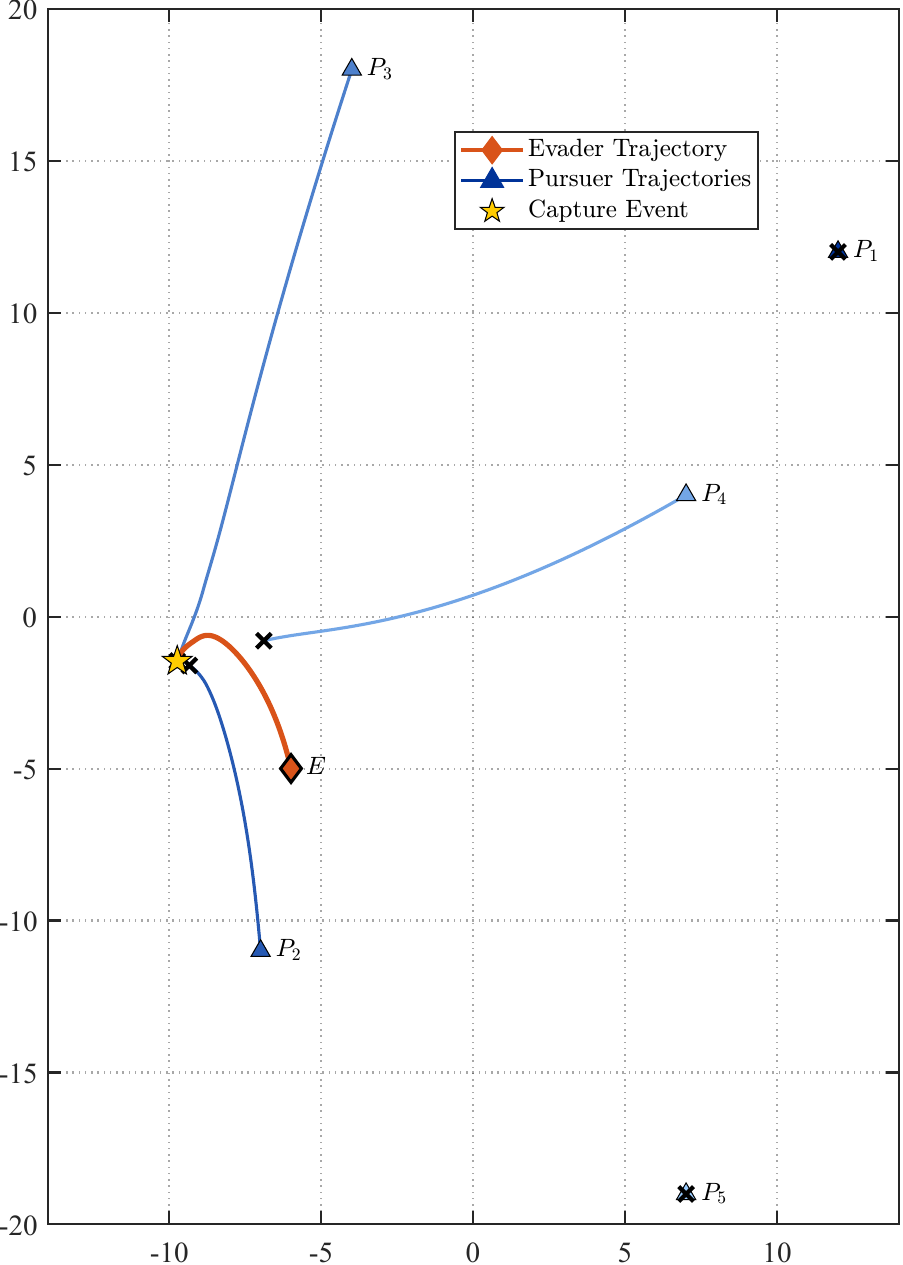}
\caption{Visualizing the Nash equilibrium for the same example scenario shown in Figure \ref{fig:safe_set}. Whenever a pursuer $i$ at time $t$ is inactive ($i \notin I(t)$), the pursuer is stationary; otherwise all agents follow their optimal headings according to Lemma \ref{Nash_equilibrium}.}
\label{fig:nash_equilibrium}
\end{figure}

\section{Conclusion}\label{sec:conclusion}
We generalized the area-based control framework to multi-agent pursuit-evasion games with heterogeneous pursuer speeds, formulating the conflict as a zero-sum game on the area of the evader's safe-reachable set, $\mathcal{S}_e(t)$. Our central contribution is the analytical derivation of the required area gradients. By applying the Leibniz integral rule, we derived these gradients directly from the geometry of the safe set's boundary, bypassing the need for a closed-form expression for the complex region formed by intersecting Apollonius discs, as shown in Figure \ref{fig:safe_set}. The resulting Nash equilibrium strategies, depicted in Figure \ref{fig:nash_equilibrium}, are decentralized, optimal, and geometrically intuitive: each active pursuer moves toward the centroid of its contributed boundary arc.

This work establishes a foundational control strategy for a broad class of pursuit-evasion scenarios. Future research can build upon this framework to incorporate additional real-world complexities, such as environmental obstacles, three-dimensional engagements, and scenarios involving partial information or uncertainty.

\bibliography{aaai2026}

\section{Technical Appendix}
\subsection{Proof of Lemma \ref{lemma_Boundary_Arc_Computation}}\label{Appendix_proof_of_Boundary_Arc_Computation}
\begin{proof}
An arc from an active circle $\partial\mathcal{B}_i$ contributes to the boundary $\partial\mathcal{S}_e$ if and only if it lies within every other active disc $\mathcal{B}_j$ (for $j \in I, j \neq i$). We'll derive the angular range for the arc on $\partial\mathcal{B}_i$ that satisfies the constraint imposed by a single disc $\mathcal{B}_j$, and then find the intersection of these ranges for all $j \neq i$. For brevity, the time dependence $(t)$ is omitted from the notation in this proof.

A point $\mathbf{q}_i$ on the circle $\partial\mathcal{B}_i$ can be parametrized by an angle $\theta$ relative to its center $\mathbf{c}_i$:
$$\mathbf{q}_i(\theta) = \mathbf{c}_i + r_i(\cos\theta\,\hat{\mathbf{x}} + \sin\theta\,\hat{\mathbf{y}}).$$
For this point to lie inside the disc $\mathcal{B}_j$, it must satisfy the condition $\|\mathbf{q}_i(\theta) - \mathbf{c}_j\|^2 \le r_j^2$. Let's analyze the term on the left-hand side:
$$
\begin{aligned}
    \| & \mathbf{q}_i(\theta) - \mathbf{c}_j\|^2 \\
    &= \|(\mathbf{c}_i - \mathbf{c}_j) + r_i(\cos\theta\,\hat{\mathbf{x}} + \sin\theta\,\hat{\mathbf{y}})\|^2 \\
    &= \|\mathbf{c}_i - \mathbf{c}_j\|^2 + r_i^2 + 2r_i(\mathbf{c}_i - \mathbf{c}_j) \cdot (\cos\theta\,\hat{\mathbf{x}} + \sin\theta\,\hat{\mathbf{y}}).
\end{aligned}
$$Let $d_{ij} = \|\mathbf{c}_j - \mathbf{c}_i\|$ be the distance between the centers, and let $\phi_{ij}$ be the angle of the vector from $\mathbf{c}_i$ to $\mathbf{c}_j$, such that $\mathbf{c}_j - \mathbf{c}_i = d_{ij}(\cos\phi_{ij}\,\hat{\mathbf{x}} + \sin\phi_{ij}\,\hat{\mathbf{y}})$. Substituting this gives:$$
\begin{aligned}
    \| & \mathbf{q}_i(\theta) - \mathbf{c}_j\|^2 \\
    &= d_{ij}^2 + r_i^2 - 2r_i(\mathbf{c}_j - \mathbf{c}_i) \cdot (\cos\theta\,\hat{\mathbf{x}} + \sin\theta\,\hat{\mathbf{y}}) \\
    &= d_{ij}^2 + r_i^2 - 2r_id_{ij}(\cos\phi_{ij}\cos\theta + \sin\phi_{ij}\sin\theta).
\end{aligned}
$$Using the identity $\cos(\theta - \phi_{ij}) = \cos\theta\cos\phi_{ij} + \sin\theta\sin\phi_{ij}$, this simplifies to:$$
\|\mathbf{q}_i(\theta) - \mathbf{c}_j\|^2 = d_{ij}^2 + r_i^2 - 2r_id_{ij}\cos(\theta - \phi_{ij}).
$$
Now, imposing the constraint $\|\mathbf{q}_i(\theta) - \mathbf{c}_j\|^2 \le r_j^2$:
$$d_{ij}^2 + r_i^2 - 2r_id_{ij}\cos(\theta - \phi_{ij}) \le r_j^2.$$
Rearranging to solve for $\cos(\theta - \phi_{ij})$ yields:
$$\cos(\theta - \phi_{ij}) \ge \frac{d_{ij}^2 + r_i^2 - r_j^2}{2r_id_{ij}}.$$
The right-hand side is exactly the term defined as $\cos(\alpha_{ij})$ in the lemma statement. The inequality becomes $\cos(\theta - \phi_{ij}) \ge \cos(\alpha_{ij})$, where $\alpha_{ij} \in [0, \pi)$. This holds for angles $\theta$ in the closed interval centered at $\phi_{ij}$:
$$\theta \in [\phi_{ij} - \alpha_{ij}, \phi_{ij} + \alpha_{ij}] \quad (\text{mod } 2\pi).$$
This interval, which we denote $\Gamma_{ij}$, is the set of angles on circle $i$ that lie inside circle $j$. For the point $\mathbf{q}_i(\theta)$ to be on the boundary of the safe-reachable set $\partial\mathcal{S}_e$, it must satisfy this condition for all $j \in I$ with $j \neq i$. Therefore, the total valid angular range for circle $i$, denoted $\Theta_i$, is the intersection of all these individual constraint intervals:
$$\Theta_i = \bigcap_{j \in I, j \neq i} \Gamma_{ij}.$$
Since each $\Gamma_{ij}$ is a compact set (a closed interval on the circle), their intersection $\Theta_i$ is also a compact set. Therefore this intersection, if non-empty, will form a single continuous interval $[\theta_i^{\min}, \theta_i^{\max}]$. This completes the proof. 
\end{proof}

\subsection{Proof of Proposition \ref{prop:level_set_gradients}}\label{appendix:level_set_gradients}
\begin{proof}
Using the chain rule, the gradient of $F_i$ with respect to any agent position vector $\mathbf{z} \in \{\mathbf{p}_1, \dots, \mathbf{p}_N, \mathbf{e}\}$ is:
\begin{equation*}
    \nabla_{\mathbf{z}} F_i = (\mathbf{J}_{\mathbf{z}} \mathbf{c}_i)^\top \nabla_{\mathbf{c}_i} F_i + \frac{\partial F_i}{\partial r_i} \nabla_{\mathbf{z}} r_i,
\end{equation*}
where $\mathbf{J}_{\mathbf{z}} \mathbf{c}_i$ is the Jacobian of $\mathbf{c}_i(\mathbf{z})$. The partial derivatives of $F_i$ with respect to its direct arguments, $\mathbf{c}_i$ and $r_i$, are straightforward:
\begin{equation*}
    \nabla_{\mathbf{c}_i} F_i = -2(\mathbf{q} - \mathbf{c}_i), \qquad \frac{\partial F_i}{\partial r_i} = -2r_i.
\end{equation*}
The derivatives of the Apollonius parameters $\mathbf{c}_i$ and $r_i$ with respect to the agent positions are also required. From Eq.~\eqref{eq:apollonius_i}, we have:
\begin{align*}
    \mathbf{c}_i &= \frac{1}{1-\alpha_i^2}\mathbf{e} - \frac{\alpha_i^2}{1-\alpha_i^2}\mathbf{p}_i, \\
    r_i &= \frac{\alpha_i}{1-\alpha_i^2}\|\mathbf{e} - \mathbf{p}_i\|.
\end{align*}
The gradients of these parameters with respect to $\mathbf{p}_i$ and $\mathbf{e}$ are:
\begin{align*}
    \mathbf{J}_{\mathbf{p}_i} \mathbf{c}_i &= -\frac{\alpha_i^2}{1-\alpha_i^2}\mathbf{I}, & \mathbf{J}_{\mathbf{e}} \mathbf{c}_i &= \frac{1}{1-\alpha_i^2}\mathbf{I}, \\
    \nabla_{\mathbf{p}_i} r_i &= -\frac{\alpha_i}{1-\alpha_i^2} \frac{\mathbf{e} - \mathbf{p}_i}{\|\mathbf{e} - \mathbf{p}_i\|}, & \nabla_{\mathbf{e}} r_i &= \frac{\alpha_i}{1-\alpha_i^2} \frac{\mathbf{e} - \mathbf{p}_i}{\|\mathbf{e} - \mathbf{p}_i\|}.
\end{align*}
For any pursuer $k \neq i$, $\mathbf{c}_i$ and $r_i$ are independent of $\mathbf{p}_k$, so $\nabla_{\mathbf{p}_k} \mathbf{c}_i = \mathbf{0}$ and $\nabla_{\mathbf{p}_k} r_i = \mathbf{0}$.

Combining these results, the gradient of $F_i$ with respect to $\mathbf{p}_i$ is:
\begin{align*}
    \nabla_{\mathbf{p}_i} F_i &= \left(-\frac{\alpha_i^2}{1-\alpha_i^2}\mathbf{I}\right)^\top \left(-2(\mathbf{q} - \mathbf{c}_i)\right) \dots \\
    &+ (-2r_i) \left(-\frac{\alpha_i}{1-\alpha_i^2} \frac{\mathbf{e} - \mathbf{p}_i}{\|\mathbf{e} - \mathbf{p}_i\|}\right) \\
    &= \frac{2\alpha_i^2}{1-\alpha_i^2}(\mathbf{q} - \mathbf{c}_i) + \frac{2\alpha_i r_i}{1-\alpha_i^2} \frac{\mathbf{e} - \mathbf{p}_i}{\|\mathbf{e} - \mathbf{p}_i\|}.
\end{align*}
Substituting the expressions for $\mathbf{c}_i$ and $r_i$ to simplify:
\begin{align*}
    \nabla_{\mathbf{p}_i} F_i &= \frac{2\alpha_i^2}{1-\alpha_i^2}\left(\mathbf{q} - \frac{\mathbf{e} - \alpha_i^2\mathbf{p}_i}{1-\alpha_i^2}\right) + \frac{2\alpha_i}{1-\alpha_i^2}\left(\frac{\alpha_i}{1-\alpha_i^2}\right)(\mathbf{e} - \mathbf{p}_i) \\
    &= \frac{2\alpha_i^2}{(1-\alpha_i^2)^2} \left( (1-\alpha_i^2)\mathbf{q} - \mathbf{e} + \alpha_i^2\mathbf{p}_i + (\mathbf{e} - \mathbf{p}_i) \right) \\
    &= \frac{2\alpha_i^2}{(1-\alpha_i^2)^2} \left( (1-\alpha_i^2)\mathbf{q} - (1-\alpha_i^2)\mathbf{p}_i \right) = \frac{2\alpha_i^2}{1-\alpha_i^2}(\mathbf{q} - \mathbf{p}_i).
\end{align*}
A similar derivation for the gradient with respect to $\mathbf{e}$ yields:
\begin{align*}
    \nabla_{\mathbf{e}} F_i &= \left(\frac{1}{1-\alpha_i^2}\mathbf{I}\right)^\top \left(-2(\mathbf{q} - \mathbf{c}_i)\right) + (-2r_i) \left(\frac{\alpha_i}{1-\alpha_i^2} \frac{(\mathbf{e} - \mathbf{p}_i)}{\|\mathbf{e} - \mathbf{p}_i\|}\right) \\
    &= -\frac{2}{1-\alpha_i^2}(\mathbf{q} - \mathbf{c}_i) - \frac{2\alpha_i r_i}{1-\alpha_i^2} \frac{\mathbf{e} - \mathbf{p}_i}{\|\mathbf{e} - \mathbf{p}_i\|} = \frac{2}{1-\alpha_i^2}(\mathbf{e} - \mathbf{q}).
\end{align*}
Finally, the gradient with respect to the spatial variable $\mathbf{q}$ is simply $\nabla_{\mathbf{q}} F_i = 2(\mathbf{q} - \mathbf{c}_i)$.
\end{proof}

\subsection{Proof of the Arc Centroid Formula}\label{app:centroid_proof}
\begin{proof}
Here, we provide a detailed proof for the explicit formula of the arc centroid $\mathbf{C}_{\text{arc},i}(t)$ presented in Theorem~\ref{thm:area_gradients}.

The centroid of a curve is defined as the vector average of the positions of all points on that curve. For an active arc $i \in I(t)$, which is a segment of a circle, the centroid is given by the line integral:
\begin{equation*}
    \mathbf{C}_{\text{arc},i}(t) \triangleq \frac{1}{L_i(t)} \int_{\text{arc}_i} \mathbf{q}(s) \, ds,
\end{equation*}
where $L_i(t)$ is the length of the arc and $\mathbf{q}(s)$ is the position vector of a point on the arc, parametrized by arc length $s$.

To evaluate this integral, we parameterize the arc using the angle $\theta$. A point $\mathbf{q}(\theta)$ on the circle of pursuer $i$ is given by:
\begin{equation*}
    \mathbf{q}(\theta) = \mathbf{c}_i(t) + r_i(t) \begin{bmatrix} \cos\theta \\ \sin\theta \end{bmatrix},
\end{equation*}
where $\mathbf{c}_i(t)$ is the center and $r_i(t)$ is the radius of the circle. The differential arc length is $ds = r_i(t) d\theta$. The total arc length is $L_i(t) = r_i(t) \Delta\theta_i$, where $\Delta\theta_i = \theta_i^{\max}(t) - \theta_i^{\min}(t)$ is the angular width of the arc.

Substituting the parametrization into the centroid definition gives:
\begin{align*}
    \mathbf{C}_{\text{arc},i}(t) &= \frac{1}{r_i \Delta\theta_i} \int_{\theta_i^{\min}}^{\theta_i^{\max}} \left( \mathbf{c}_i + r_i \begin{bmatrix} \cos\theta \\ \sin\theta \end{bmatrix} \right) r_i \,d\theta \\
    &= \frac{1}{\Delta\theta_i} \int_{\theta_i^{\min}}^{\theta_i^{\max}} \left( \mathbf{c}_i + r_i \begin{bmatrix} \cos\theta \\ \sin\theta \end{bmatrix} \right) d\theta.
\end{align*}
We can split the integral into two parts:
\begin{align*}
    \mathbf{C}_{\text{arc},i}(t) = \frac{1}{\Delta\theta_i} \left( \int_{\theta_i^{\min}}^{\theta_i^{\max}} \mathbf{c}_i \,d\theta + r_i \int_{\theta_i^{\min}}^{\theta_i^{\max}} \begin{bmatrix} \cos\theta \\ \sin\theta \end{bmatrix} d\theta \right).
\end{align*}
The first integral evaluates to $\mathbf{c}_i \Delta\theta_i$. The second integral is:
\begin{align*}
    \int_{\theta_i^{\min}}^{\theta_i^{\max}} \begin{bmatrix} \cos\theta \\ \sin\theta \end{bmatrix} d\theta &= \begin{bmatrix} \sin\theta \\ -\cos\theta \end{bmatrix}_{\theta_i^{\min}}^{\theta_i^{\max}} \\
    &= \begin{bmatrix} \sin(\theta_i^{\max}) - \sin(\theta_i^{\min}) \\ \cos(\theta_i^{\min}) - \cos(\theta_i^{\max}) \end{bmatrix}.
\end{align*}
Combining these results, we have:
\begin{equation*}
    \mathbf{C}_{\text{arc},i}(t) = \mathbf{c}_i + \frac{r_i}{\Delta\theta_i} \begin{bmatrix} \sin(\theta_i^{\max}) - \sin(\theta_i^{\min}) \\ \cos(\theta_i^{\min}) - \cos(\theta_i^{\max}) \end{bmatrix}.
\end{equation*}
To simplify the vector component, we use the trigonometric sum-to-product identities:
\begin{align*}
    \sin(A) - \sin(B) &= 2 \cos\left(\frac{A+B}{2}\right) \sin\left(\frac{A-B}{2}\right) \\
    \cos(B) - \cos(A) &= 2 \sin\left(\frac{A+B}{2}\right) \sin\left(\frac{A-B}{2}\right).
\end{align*}
Let $m_i = (\theta_i^{\max} + \theta_i^{\min})/2$ be the midpoint angle and recall $\Delta\theta_i = \theta_i^{\max} - \theta_i^{\min}$. The vector becomes:
\begin{align*}
    &\begin{bmatrix} 2\cos(m_i)\sin(\Delta\theta_i/2) \\ 2\sin(m_i)\sin(\Delta\theta_i/2) \end{bmatrix} = 2\sin(\frac{\Delta\theta_i}{2}) \begin{bmatrix} \cos(m_i) \\ \sin(m_i) \end{bmatrix}.
\end{align*}
This resulting vector has a clear geometric interpretation. The unit vector $\mathbf{u}_i(t) = [\cos(m_i), \sin(m_i)]^\top$ points from the circle's center $\mathbf{c}_i$ towards the midpoint of the chord connecting the arc's endpoints. The scalar term $2\sin(\Delta\theta_i/2)$ relates to the chord length, $\ell_i(t)$. The chord length is given by $\ell_i(t) = 2r_i \sin(\Delta\theta_i/2)$. Thus, we can write $2\sin(\Delta\theta_i/2) = \ell_i(t)/r_i(t)$.

Substituting these geometric forms back into the expression for the centroid yields:
\begin{align*}
    \mathbf{C}_{\text{arc},i}(t) &= \mathbf{c}_i + \frac{r_i}{\Delta\theta_i} \left( \frac{\ell_i}{r_i} \mathbf{u}_i \right) \\
    &= \mathbf{c}_i + \frac{\ell_i}{\Delta\theta_i} \mathbf{u}_i.
\end{align*}
Finally, by substituting $\Delta\theta_i = L_i/r_i$, we arrive at the desired compact expression for the arc centroid:
\begin{equation*}
    \mathbf{C}_{\text{arc},i}(t) = \mathbf{c}_i(t) + r_i(t)\left( \frac{\ell_i(t)}{L_i(t)} \right) \mathbf{u}_i(t).
\end{equation*}
\end{proof}

\end{document}